\documentclass[twocolumn,amsthm]{autartx}
\usepackage{amssymb,amsmath}
\usepackage{comment}
\usepackage{newtxtext,newtxmath}
\usepackage{amsfonts}
\usepackage{tikz}
  \usetikzlibrary{arrows,automata,shapes}
  \usetikzlibrary{backgrounds}

\usepackage[textwidth=1.9cm,color=green!10,textsize=footnotesize]{todonotes}

\newcommand{\eps}{\varepsilon}
\DeclareMathOperator{\M}{\mathcal M}

\begin{document}

\begin{frontmatter}
\title{Complexity of Detectability, Opacity and A-Diagnosability\\ for Modular Discrete Event Systems}

\author[tm]{Tom{\' a}{\v s}~Masopust}\ead{masopust{@}math.cas.cz}
\and
\author[xy]{Xiang Yin}\ead{xiangyin@umich.edu}

\address[tm]{Institute of Mathematics, Czech Academy of Sciences, {\v Z}i{\v z}kova 22, 616 62 Brno, Czechia}
\address[xy]{Department of Automation, Shanghai Jiao Tong University, Shanghai 200240, China}


\begin{abstract}
  We study the complexity of deciding whether a modular discrete event system is detectable (resp. opaque, A-diagnosable).
  Detectability arises in the state estimation of discrete event systems,
  opacity is related to the privacy and security analysis,
  and A-diagnosability appears in the fault diagnosis of stochastic discrete event systems.
  Previously, deciding weak detectability (opacity, A-diagnosability) for monolithic systems was shown to be PSPACE-complete. In this paper, we study the complexity of deciding weak detectability (opacity, A-diagnosability) for modular systems. We show that the complexities of these problems are significantly worse than in the monolithic case. Namely, we show that deciding modular weak detectability (opacity, A-diagnosability) is EXPSPACE-complete.
  We further discuss a special case where all unobservable events are private, and show that in this case the problems are PSPACE-complete. Consequently, if the systems are all fully observable, then deciding weak detectability (opacity) for modular systems is PSPACE-complete.
\end{abstract}

\end{frontmatter}

\section{Introduction}
  Discrete event systems (DES) are dynamical systems with discrete state-spaces and event-triggered dynamics.
  In most real-world applications,
  DES models are obtained by several local modules, modeled as finite-state automata, running synchronously.  This leads to the research on modular DES.
  In the study of modular DES, the main challenge is the state-space exploration problem -- the number of states in the monolithic model grows exponentially fast as the  number of local modules increases.
  Understanding the computational complexity therefore becomes the essential problem in the analysis of modular DES, which has drawn a considerable attention in the literature.
  For example, Gohari and Wonham~\cite{gohari2000complexity} and Rohloff and Lafortune~\cite{rohloff2005pspace} investigated the complexities of standard supervisory control problems for modular DES.
  The complexities of diagnosability, detectability and predictability for modular DES have also been investigated~\cite{YinLafortune17} as well as the complexity of verifying nonblockingness for modular DES~\cite{Masopust2017Complexity}.
  It turns out that many problems tractable in the monolithic setting become intractable in the modular setting due to the state-space exploration issue.
  There are many results in the literature on finding tractable solutions for problems in the context of modular DES~\cite{feng2008supervisory,gummadi2011tractable,hill2010multi,komenda2014coordination,leduc2005hierarchical,saboori2010reduced,schmidt2012efficient}.

  Detectability, opacity and diagnosability are important system-theoretic properties of discrete event systems.
  Detectability arises in the state estimation of DES.
  In particular, weak detectability asks whether the current state of the system can be determined unambiguously after a finite number of observations via some trajectory~\cite{ShuLinYing2007}.
  Opacity is a property related to the privacy and security analysis of DES~\cite{jacob2016overview}. The system has a secret modeled as a set of states and an intruder is modeled as a passive observer with a limited observation. The system is opaque if the intruder never knows for sure that the system is in a secret state, that is, the system never reveals its secret.
  Fault diagnosis is another important task in DES. Several different notions of diagnosability have been proposed in the literature~\cite{zaytoon2013overview}. For example, the notion of diagnosability of Sampath et al.~\cite{sampath1995diagnosability} requires that an occurrence of a fault can always be detected within a finite delay. Thorsley and Teneketzis~\cite{thorsley2005diagnosability} proposed a weaker version of diagnosability, called A-diagnosability. Compared to diagnosability, where the fault has to be detected on every path within a finite delay, A-diagnosability requires that there always exists a possibility to detect the fault event, and hence the probability of detection goes to one as the length of the trajectory increase.

  It has been shown in the literature that deciding weak detectability~\cite{Zhang17,Masopust2017Detectability}, opacity~\cite{CassezDM12} and A-diagnosability~\cite{bertrand2014foundation,chenrevised} are PSPACE-complete problems for monolithic systems.
  However, these properties have not yet been fully investigated for modular systems.
  A PSPACE lower-bound has been shown for weak detectability in the modular setting~\cite{YinLafortune17} and a question remains whether the bound is tight or not.
  Regarding opacity in the context of modular DES, Saboori and Hadjicostis~\cite{saboori2010reduced} proposed an exponential-time algorithm for the verification of modular opacity under certain restrictive assumptions.
  However, the precise complexity for deciding opacity is still open for modular DES in the general case without any restrictive assumptions.
  Moreover, to the best of our knowledge, A-diagnosability has only been investigated for monolithic DES and its complexity in the modular setting is open.

  In this paper, we investigate the computational complexity of deciding detectability, opacity and A-diagnosability for modular discrete event systems.
  Our contributions are as follows.
  We first show that deciding any of the properties is EXPSPACE-complete for modular systems, and hence significantly more complex than their monolithic counterparts.
  Consequently, there is, in general, neither a polynomial-time nor a polynomial-space algorithm verifying these properties.
  Previously, only weak detectability has been shown to be PSPACE-hard for modular systems~\cite{YinLafortune17} and the complexities of opacity and A-diagnosability have never been considered in the literature.
  Our results provide complete complexity results for these three properties.
  Furthermore, we  investigate a special case where all unobservable events are private.  In this case, we show that the problems are easier, namely weak detectability, opacity and A-diagnosability are PSPACE-complete, and hence solvable by polynomial-space algorithms.
  Our complexity results are obtained in a uniform manner in the sense that a similar construction is used to show hardness of all the properties.
  We believe that our results bring new insight into the relationship among detectability, opacity and A-diagnosability from the computational complexity point of view.

\section{Preliminaries and Definitions}
  For a set $A$, $|A|$ denotes the cardinality of $A$, and $2^{A}$ the power set of $A$. An {\em alphabet\/} $\Sigma$ is a finite nonempty set with elements called {\em events}. A {\em word\/} over $\Sigma$ is a sequence of events of $\Sigma$. Let $\Sigma^*$ denote the set of all finite words over $\Sigma$; the {\em empty word\/} is denoted by $\varepsilon$. For a word $u \in \Sigma^*$, $|u|$ denotes its length. As usual, the notation $\Sigma^+$ stands for $\Sigma^*\setminus\{\varepsilon\}$.

  A {\em nondeterministic finite automaton\/} (NFA) over an alphabet $\Sigma$ is a structure $A = (Q,\Sigma,\delta,I,F)$, where $Q$ is a finite nonempty set of states, $I\subseteq Q$ is a nonempty set of initial states, $F \subseteq Q$ is a set of marked states, and $\delta \colon Q\times\Sigma \to 2^Q$ is a transition function that can be extended to the domain $2^Q\times\Sigma^*$ by induction. The {\em language generated by $A$\/} is the set $L(A) = \{w\in \Sigma^* \mid \delta(I,w)\neq\emptyset\}$ and
  the {\em language recognized by $A$\/} is the set $L_m(A) = \{w\in \Sigma^* \mid \delta(I,w)\cap F \neq\emptyset\}$. Equivalently, the transition function $\delta$ can be seen as a relation $\delta \subseteq Q\times \Sigma \times Q$.

  A {\em discrete event system\/} (DES) is modeled as an NFA $G$ with all states marked. Therefore we simply write $G=(Q,\Sigma,\delta,I)$ without specifying the set of marked states. The alphabet $\Sigma$ is partitioned into two disjoint subsets $\Sigma_o$ and $\Sigma_{uo}=\Sigma\setminus\Sigma_o$, where $\Sigma_o$ is the set of {\em observable events\/} and $\Sigma_{uo}$ the set of {\em unobservable events}.

  The problems studied in this paper are based on the observation of events. The observation is described by a projection $P\colon \Sigma^* \to \Sigma_o^*$. The {\em projection} $P\colon \Sigma^* \to \Sigma_o^*$ is a morphism defined by $P(a) = \varepsilon$ for $a\in \Sigma\setminus \Sigma_o$, and $P(a)= a$ for $a\in \Sigma_o$. The action of $P$ on a word $w=\sigma_1\sigma_2\cdots\sigma_n$ with $\sigma_i \in \Sigma$ for $1\le i\le n$ is to erase all events from $w$ that do not belong to $\Sigma_o$; namely, $P(\sigma_1\sigma_2\cdots\sigma_n)=P(\sigma_1) P(\sigma_2) \cdots P(\sigma_n)$. The definition can readily be extended to infinite words and languages.

  Let $G=(Q,\Sigma,\delta,I)$ be a discrete event system, and let $P$ be a projection from $\Sigma$ to $\Delta\subseteq\Sigma$. We use the notation $P(G)$ to denote the DES $P(G)=(Q,\Delta,\delta',I)$, where the transition function $\delta' = \{ (p,P(a),q) \mid (p,a,q)\in\delta \}$. Intuitively, $P(G)$ has the same structure as $G$ with unobservable transitions labeled with $\eps$.

  As usual when partially-observed DES are studied~\cite{sampath1995diagnosability,ShuLin2011}, we make the following two assumptions on the DES $G=(Q,\Sigma,\delta,I)$:
    (1) $G$ is {\em deadlock free}, that is, for every state of the system, at least one event can occur. Formally, for every $q\in Q$, there is $\sigma \in \Sigma$ such that $\delta(q,\sigma)\neq\emptyset$.
    (2) No loop in $G$ consists solely of unobservable events: for every $q\in Q$ and every $w \in \Sigma_{uo}^+$, $q\notin \delta(q,w)$.

  In many applications, DES $G$ is obtained by the parallel composition of a set of local modules $\{G_1,G_2,\dots,G_n\}$, where $G_i=(Q_i,\Sigma_i,\delta_i,I_i)$.
  That is, $G=G_1\|G_2\|\cdots\|G_n$,
  where ``$\|$" denotes the parallel composition operator~\cite[p. 78]{Lbook}.

  A {\em decision problem\/} is a yes-no question, such as ``Is the language $L_m(A)$ of an automaton $A$ empty?'' A decision problem is {\em decidable\/} if there exists an algorithm solving the problem. Complexity theory classifies decidable problems into classes based on the time or space an algorithm needs to solve the problem. The complexity classes we consider in this paper are PSPACE and EXPSPACE denoting the classes of problems solvable by a deterministic polynomial-space algorithm and by a deterministic exponential-space algorithm, respectively. A decision problem is PSPACE-complete (resp. EXPSPACE-complete) if it belongs to PSPACE (resp. EXPSPACE) and every problem from PSPACE (resp. EXPSPACE) can be reduced to it by a deterministic polynomial-time algorithm. By the space hierarchy theorem~\cite{StearnsHL65}, we know that PSPACE is a strict subclass of EXPSPACE. Thus, if a problem is EXPSPACE-complete, there is no polynomial-space and hence no polynomial-time algorithm solving the problem.

\section{Modular Detectability}
  In this section, we investigate the complexity of deciding modular detectability.
  First, we recall the definitions of two basic variants of detectability~\cite{ShuLinYing2007} and then we define their modular counterparts that we investigate in this paper.

  Still, let $\Sigma$ be an alphabet, $\Sigma_o\subseteq\Sigma$ be the set of observable events, and $P$ be the projection from $\Sigma$ to $\Sigma_o$. Let $\mathbb{N}$ denote the set of all natural numbers.

  The set of infinite sequences of events generated by a DES $G$ is denoted by $L^\omega (G)$. For $w \in L^\omega (G)$, we denote the set of its prefixes by $Pr(w)$.

  \begin{defn}[Strong detectability]
    A discrete event system $G=(Q,\Sigma,\delta,I)$ is strongly detectable with respect to $\Sigma_{uo}$ if we can determine, after a finite number of observations, the current and subsequent states of the system for all trajectories of the system, that is,
      $(\exists n \in \mathbb{N}) (\forall s \in L^\omega(G))(\forall t \in Pr(s))
      [|P(t)| > n \Rightarrow |R_G(t)| = 1]$,
    where $R_G(t) = \{x \in Q \mid \exists\, t' \in L(G) \textrm{ such that } P(t) = P(t') \text{ and } x \in \delta(I,t')\}$.
  \end{defn}

  \begin{defn}[Strong periodic detectability]
    A discrete event system $G=(Q,\Sigma,\delta,I)$ is strongly periodically detectable with respect to $\Sigma_{uo}$ if we can periodically determine the current state of the system for all trajectories of the system, that is,
    $(\exists n \in \mathbb{N}) (\forall s \in L^\omega(G))(\forall t \in Pr(s))(\exists t' \in \Sigma^*)
     [tt'\in Pr(s) \land |P(t')| < n \land |R_G(tt')| = 1]$.
  \end{defn}

  The modular version of the problems is defined as follows.
  \begin{defn}[Strong (periodic) modular detectability]
    Given a set of discrete event systems $\{G_1,G_2,\ldots,G_n\}$ and a set of unobservable events $\Sigma_{uo}$. The strong (periodic) modular detectability problem asks whether the discrete event system $G_1\|G_2\| \cdots \| G_n$ is strongly (periodically) detectable with respect to $\Sigma_{uo}$.
  \end{defn}

  Deciding strong modular detectability is a PSPACE-hard problem~\cite{YinLafortune17}.
  Our first result improves this lower bound by showing that polynomial space is sufficient to solve the problem.

  \begin{thm}
    Deciding strong (periodic) modular detectability is a PSPACE-complete problem.
  \end{thm}
  \begin{proof}
    PSPACE-hardness is known~\cite{YinLafortune17}, and the proof also shows PSPACE-hardness of deciding strong periodic modular detectability.

    To show membership in PSPACE, we adjust the polynomial detector $G_{det}$ of Shu and Lin~\cite{ShuLin2011}. For a DES $G$, Shu and Lin construct an NFA $G_{det}$ of polynomial size whose states are subsets of states of $G$ of cardinality one or two\footnote{Shu and Lin construct the automaton $G_{det}$ with a single initial state that is a subset of the set of states not necessarily of cardinality one or two. However, their construction can easily be modified so that we have a set of initial states each of cardinality at most two, and hence the cardinality condition is satisfied by all states.}, such that $G$ is strongly detectable if and only if
    (a) any state reachable from any loop in $G_{det}$ is of cardinality one;
    and strongly periodically detectable if and only if
    (b) all loops in $G_{det}$ include at least one state of cardinality one.

    Let $G=\|_{i=1}^{n} G_i$ be a modular system. The state set of $G$ consists of $n$-tuples of states of $G_i$. Constructing the detector $G_{det}$ using the construction of Shu and Lin~\cite{ShuLin2011} results in an automaton where every state contains either a single $n$-tuple or two $n$-tuples. Thus, a PSPACE algorithm can store a state of $G_{det}$ in polynomial space and use the nondeterministic search to verify that the property (a) (resp. (b)) is not satisfied. Since PSPACE is closed under complement, there is a PSPACE algorithm that can verify that (a) (resp. (b)) are satisfied, which completes the proof.
  \end{proof}

  Notice that if $G$ is a parallel composition of a constant number of systems, that is, $n\le k$ for some constant $k$, then the problem is NL-complete; NL is the class of problems decidable by a nondeterministic logarithmic-space algorithm. This follows from the fact that every $n$-tuple is then of constant length and requires thus only a logarithmic space to store the constant number of states in binary. Recall that the space hierarchy theorem~\cite{StearnsHL65} implies that NL is a strict subclass of PSPACE, and hence the strong modular detectability problem is significantly simpler if the number of systems is bounded a priori.

  Strong detectability requires that one can always determine the current state of the system unambiguously after a finite number of observations. A weaker version of detectability, so-called weak detectability~\cite{ShuLinYing2007}, requires that one can determine the current state of the system for some trajectory.

  \begin{defn}[Weak detectability]
    A discrete event system $G=(Q,\Sigma,\delta,I)$ is weakly detectable with respect to $\Sigma_{uo}$ if  we can determine, after a finite number of observations, the current and subsequent states of the system for some trajectories of the system, that is,
    $(\exists n \in \mathbb{N}) (\exists s \in L^\omega(G))(\forall t \in Pr(s))
      [ |P(t)| > n \Rightarrow |R_G(t)| = 1]$.
  \end{defn}

  \begin{defn}[Weak periodic detectability]
    A discrete event system $G=(Q,\Sigma,\delta,I)$ is weakly periodically detectable with respect to $\Sigma_{uo}$ if we can periodically determine the current state of the system for some trajectories of the system, that is,
    $(\exists n \in \mathbb{N}) (\exists s \in L^\omega(G))(\forall t \in Pr(s))(\exists t' \in \Sigma^*)
      [ tt'\in Pr(s) \land |P(t')| < n \land |R_G(tt')| = 1]$.
  \end{defn}

  Similarly as strong (periodic) modular detectability we define weak (periodic) modular detectability.
  \begin{defn}[Weak (periodic) modular detectability]
    Given a set of discrete event systems $\{G_1,G_2,\ldots,G_n\}$ and a set of unobservable events $\Sigma_{uo}$. The weak (periodic) modular detectability problem asks whether the system $G_1\|G_2\| \cdots \| G_n$ is weakly (periodically) detectable with respect to $\Sigma_{uo}$.
  \end{defn}

  Is has been shown that deciding weak modular detectability is a PSPACE-hard problem~\cite{YinLafortune17,Zhang17}. However, the precise complexity was left open. The PSPACE-hardness result does not exclude the possibility that the problem can be solvable in polynomial space. We now show that, in fact, the problem requires exponential space.

  We first need the following lemma.
  \begin{lem}\label{lem1}
    Let $\Sigma$ be an alphabet and $n\ge 1$ be a natural number. There exist $n$ six-state automata $A_i$ such that $P( L_m(\|_{i=1}^{n} A_i)) = \Sigma^{2^n-1}$, where $P$ is a projection from the overall alphabet of $\|_{i=1}^{n} A_i$ to $\Sigma$.
  \end{lem}
  \begin{proof}
    Let $\Gamma=\{a_1,a_2,\ldots,a_{n}\}$. For $i=1,\ldots,n$, we define the automaton
    $
      A_i=(\{0,1,p,q,r,s\},\Gamma,\delta_i,0,\{1\})
    $
    where
    \begin{align*}
      \delta_i
        & = \{(0,a_j,p), (p,b,0), (1,a_j,q), (q,b,1) \mid j < i,\, b\in\Sigma\} \\
        & \cup \{(0,a_i,r), (r,b,1) \mid b\in\Sigma \}  \\
        & \cup \{(1,a_j,s), (s,b,0) \mid j > i,\, b\in\Sigma\}\,.
    \end{align*}
    See Example~\ref{ex1} for an illustration.

    The idea of the construction is that the parallel composition counts from $0$ to $2^n-1$ in binary. The initial state of the composition is the state $(0,0,\ldots,0)$ representing number $0$ in binary, which is modified step by step to the state $(1,1,\ldots,1)$ representing $2^n-1$ in binary. Every odd transition under an event from $\Gamma$ is used to count the number of steps, and every even transition under an event from $\Sigma$ is to give the required language  $\Sigma^{2^n-1}$ in the projection to $\Sigma$.

    We now prove by induction on $n\ge 1$ that the parallel composition $A_n\|\cdots\|A_1$ accepts a language $L_n \subseteq (\Gamma\Sigma)^{2^n-1}$ such that the projection of $L_n$ to $\Gamma$ is a singleton.

    For $n=1$, we have a single automaton $A_1$ recognizing the language $L_1=\{a_1b \mid b\in \Sigma\}$, and hence the claim holds.

    Assume that the claim holds for $n$. We now prove it for $n+1$. In this case, $\Gamma=\{a_1,\ldots,a_{n+1}\}$. Let $w\in L_n$. Then $w$ does not contain event $a_{n+1}$. Notice that the automaton $A_{n+1}$ cycles between states $0$ and $r$ when reading $w$. Thus, in the composition $A_{n+1}\|A_n\|\cdots\|A_1$, we have
    \begin{align*}
      (0,0,\ldots,0) & \xrightarrow{w} (0,1,\ldots,1) \\
                     & \xrightarrow{w'} (1,0,\ldots,0)\\
                     & \xrightarrow{w} (1,1,\ldots,1)
    \end{align*}
    where $w'=a_{n+1}b$ for some $b\in \Sigma$, that is, the projection of $w'$ to $\Gamma$ is the singleton $\{a_{n+1}\}$. The parallel composition therefore accepts the word $ww'w$, which is of the form $(\Gamma\Sigma)^{2^{n+1}-1}$, and where the projection of $ww'w$ to $\Gamma$ is a unique word by the induction hypothesis.
  \end{proof}

  \begin{exmp}\label{ex1}
    We demonstrate the construction for $n=3$. The automaton $A_3 \| A_2 \| A_1$ is shown in Fig.~\ref{fig1}. The reason to write the automata from right to left is that the states of their parallel composition then correspond to the binary representation of numbers, see below. It can be verified that $P(L_m(\|_{i=1}^{n} A_i)) = \Sigma^{7}$. This is because the sequence of states from $(0,0,0)$ to the accepting configuration $(1,1,1)$ is
    \begin{align*}
      (0,0,0) & \xrightarrow{a_1} (p,p,r) \xrightarrow{\Sigma} (0,0,1) \\
              & \xrightarrow{a_2} (p,r,s) \xrightarrow{\Sigma} (0,1,0) \\
              & \xrightarrow{a_1} (p,q,r) \xrightarrow{\Sigma} (0,1,1) \\
              & \xrightarrow{a_3} (r,s,s) \xrightarrow{\Sigma} (1,0,0) \\
              & \xrightarrow{a_1} (q,p,r) \xrightarrow{\Sigma} (1,0,1) \\
              & \xrightarrow{a_2} (q,r,s) \xrightarrow{\Sigma} (1,1,0) \\
              & \xrightarrow{a_1} (q,q,r) \xrightarrow{\Sigma} (1,1,1)
    \end{align*}
    Here $\xrightarrow{\Sigma}$ denotes a transition $\xrightarrow{\sigma}$ for some $\sigma\in\Sigma$.\qed
  \end{exmp}

  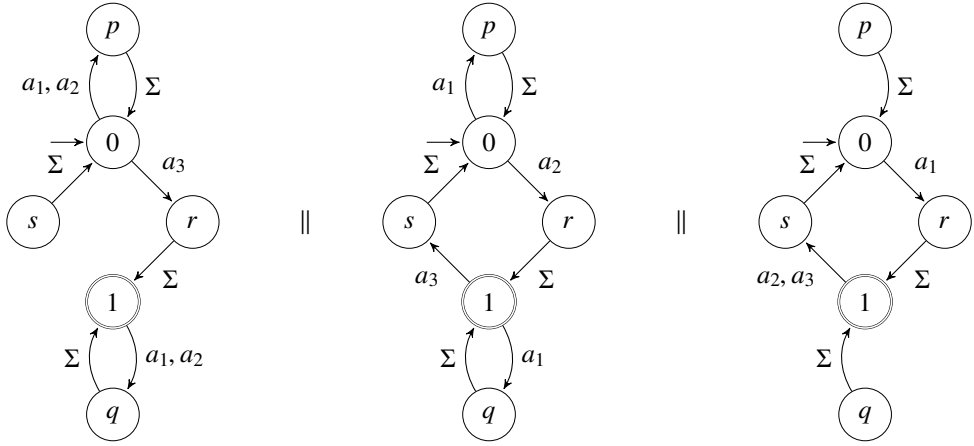
\begin{figure*}
    \centering
      \begin{tikzpicture}[>=stealth',auto,baseline,->,shorten >=1pt,node distance=1.5cm, state/.style={ellipse,minimum size=7mm,very thin,draw=black,initial text=}]
      \node[state,initial]    (10) {$0$};
      \node[state]            (1p) [above of=10] {$p$};
      \node[state]            (1r) [below right of=10] {$r$};
      \node[state,accepting]  (11) [below left of=1r] {$1$};
      \node[state]            (1q) [below of=11] {$q$};
      \node[state]            (1s) [below left of=10] {$s$};

      \node[state,initial]    (20) [left of=10,node distance=5cm] {$0$};
      \node[state]            (2p) [above of=20] {$p$};
      \node[state]            (2r) [below right of=20] {$r$};
      \node[state,accepting]  (21) [below left of=2r] {$1$};
      \node[state]            (2q) [below of=21] {$q$};
      \node[state]            (2s) [below left of=20] {$s$};

      \node[state,initial]    (30) [left of=20,node distance=5cm] {$0$};
      \node[state]            (3p) [above of=30] {$p$};
      \node[state]            (3r) [below right of=30] {$r$};
      \node[state,accepting]  (31) [below left of=3r] {$1$};
      \node[state]            (3q) [below of=31] {$q$};
      \node[state]            (3s) [below left of=30] {$s$};

      \node[] (x) [right of=3r] {$\|$};
      \node[] (x) [right of=2r] {$\|$};

      \path
        (10) edge node {$a_1$} (1r)
        (11) edge node {$a_2,a_3$} (1s)

        (20) edge[bend left] node {$a_1$} (2p)
        (21) edge[bend left] node {$a_1$} (2q)
        (20) edge node {$a_2$} (2r)
        (21) edge node {$a_3$} (2s)

        (30) edge[bend left] node {$a_1,a_2$} (3p)
        (31) edge[bend left] node {$a_1,a_2$} (3q)
        (30) edge node {$a_3$} (3r)

        (1r) edge node {$\Sigma$} (11)
        (2r) edge node {$\Sigma$} (21)
        (3r) edge node {$\Sigma$} (31)
        (1s) edge node {$\Sigma$} (10)
        (2s) edge node {$\Sigma$} (20)
        (3s) edge node {$\Sigma$} (30)
        (1p) edge[bend left] node {$\Sigma$} (10)
        (1q) edge[bend left] node {$\Sigma$} (11)
        (2p) edge[bend left] node {$\Sigma$} (20)
        (2q) edge[bend left] node {$\Sigma$} (21)
        (3p) edge[bend left] node {$\Sigma$} (30)
        (3q) edge[bend left] node {$\Sigma$} (31)
        ;
      \end{tikzpicture}
      \caption{Automaton for $A_3 \| A_2 \| A_1$}
      \label{fig1}
  \end{figure*}

  Lemma~\ref{lem1} can be modified to obtain the language $\Sigma^{2^n-k}$, for $0< k \le 2^n$. The idea is to start in the state representing number $k-1$ in binary. For instance, in Example~\ref{ex1}, we can obtain $\Sigma^{8-4}$ by starting from state $(0,1,1)$ that represents number $3$ in binary in the parallel composition $A_3 \| A_2 \| A_1$.

  We can now prove the main theorem of this section.
  \begin{thm}\label{thm1}
    Deciding weak (periodic) modular detectability is an EXPSPACE-complete problem.
  \end{thm}
  \begin{proof}
    To check weak modular detectability in exponential space, we can compute the overall system $G = \|_{i=1}^{n} G_i$ and use the method of Shu and Lin~\cite{ShuLin2011} computing the observer of $G$ by storing only the current state of the observer and using the nondeterministic search strategy. The states of the observer of $G$ are subsets of $n$-tuples, and hence every state can be of size up to $k^n$, where $k$ is the maximum number of states among all $G_i$. Using the nondeterministic search method then shows that the problem is in EXPSPACE.

    To show that deciding weak modular detectability is EXPSPACE-hard, we consider a $2^n$-space-bounded Turing machine $\M = (Q,T,I,\delta,b,q_o,q_a)$ with $b\in T$ being a blank symbol, and an input word $x=x_1x_2\cdots x_n$.
    We construct, in polynomial time, a sequence of $m$ systems $G_i$, some of which are a parallel composition of other systems, such that $\|_{i=1}^{m} G_i$ is weakly (periodically) detectable if and only if $\M$ accepts $x$.

    We first describe how to encode a computation of $\M$ on $x$. A configuration of $\M$ on $x$ consists of a current state $q\in Q$, the position $1\leq \ell\leq 2^n$ of the head, and the tape contents $\theta_1,\ldots,\theta_{2^n}$ with $\theta_i\in T$. We represent it by a sequence
    \[
      \theta_1 \cdots \theta_{\ell-1} (q,\theta_{\ell}) \theta_{\ell+1} \cdots \theta_{2^n}
    \]
    of symbols from $T \cup Q \times T$. A run of $\M$ on $x$ is represented as a word
    $\# w_1 \# w_2 \# \cdots \allowbreak \# w_m \#$, where $w_i\in (T \cup Q \times T)^{2^n}$ and
    $\#\notin T \cup Q \times T$ is a fresh separator symbol.

    Meyer and Stockmeyer~\cite{MeyerS72} show that there is a regular expression $E$ recognizing all words over $\Delta=\{\#\}\cup T \cup Q\times T$ that do not correctly encode an accepting run of $\M$ on $x$, that is,
    \begin{itemize}
      \item if $x\notin L(\M)$, then $L(E)=\Delta^*$,
      \item if $x\in L(\M)$, then $L(E)\neq\Delta^*$,
    \end{itemize}
    where $L(\M)$ denotes the set of inputs accepted by $\M$.

    The regular expression is constructed in the following steps:
    \begin{description}
      \item[(A)] Words that do not begin with the initial configuration, that is, $\#w_1\#$ is not of the form $\#(q_o,x_1) x_2 \cdots x_n b^{2^n-n}\#$:
        \begin{align}
          & ( (\Delta\setminus \#)
            \cup
            \# \cdot ((\Delta\setminus (q_o,x_1))
              \cup (q_o,x_1) \cdot ((\Delta \setminus x_2) \nonumber\\
          &     \cup x_2 \cdot ((\Delta \setminus x_3)
                  \cup \cdots (\Delta\setminus x_n))) \cdots ) \cdot \Delta^* \\
          & \cup
            \Delta^{n+1} \cdot b^* \cdot (\Delta \setminus \{b,\#\}) \cdot \Delta^* \\
          & \cup
            \# \cdot (\Delta\cup\eps)^{2^n-1} \cdot \# \cdot \Delta^* \\
          & \cup
            \# \cdot \Delta^{2^n} \cdot (\Delta \setminus \#) \cdot \Delta^*
        \end{align}

      \item[(B)] Words that do not contain the accepting state $q_a$:
        \[
          (\Delta \setminus ( \cup_{t\in T} (q_a,t)))^*
        \]

      \item[(C)] Words that are not of the form $\# w_1 \# w_2 \# \cdots \# w_k \#$, where $w_{i+1}$ is the configuration obtained from the configuration $w_i$ by the rules of $\M$:
        \[
          \bigcup_{c_1,c_2,c_3 \in \Delta} \Delta^* \cdot c_1c_2c_3 \cdot \Delta^{2^n-1} \cdot (\Delta\setminus N(c_1,c_2,c_3)) \cdot \Delta^*
        \]
        where $N(c_1,c_2,c_3) \subseteq \Delta$ is the set of symbols that could legally occupy the $j$-th symbol of the next configuration given that the $(j-1)$, $j$, and $(j+1)$-th symbols of the current configuration are $c_1,c_2,c_3$, respectively. Also $N(c_1,\#,c_3)=\{\#\}$ for every $c_1,c_3\in\Delta$.
    \end{description}

    The regular expression consists of several unions of structurally simpler regular expressions. Most of these regular expressions can easily be translated to an NFA in polynomial time by a direct transformation. For instance, an NFA for the expression \textbf{A}(2) is depicted in Fig.~\ref{fig2}.
    \begin{figure*}
      \centering
        \begin{tikzpicture}[>=stealth',auto,baseline,->,shorten >=1pt,node distance=2cm, state/.style={ellipse,minimum size=7mm,very thin,draw=black,initial text=}]
        \node[state,initial]    (0) {$0$};
        \node[state]            (1) [right of=0] {$1$};
        \node                   (2) [right of=1] {$\cdots$};
        \node[state]            (3) [right of=2] {$n+1$};
        \node[state,accepting]  (4) [right of=3,node distance=3.5cm] {$n+2$};

        \path
          (0) edge node {$\Delta$} (1)
          (1) edge node {$\Delta$} (2)
          (2) edge node {$\Delta$} (3)
          (3) edge[loop above] node {$b$} (3)
          (3) edge node {$\Delta\setminus\{b,\#\}$} (4)
          (4) edge[loop above] node {$\Delta$} (4)
          ;
        \end{tikzpicture}
        \caption{An NFA for the expression $\Delta^{n+1} \cdot b^* \cdot (\Delta \setminus \{b,\#\}) \cdot \Delta^*$}
        \label{fig2}
    \end{figure*}
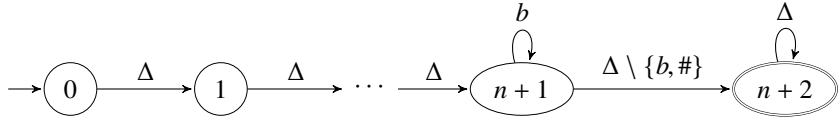
    The construction of an NFA is, however, not easy for those regular expressions that contain parts of the form $\Delta^{2^n}$, because such a direct translation would require $2^n$ transitions labeled with all symbols of $\Delta$, which cannot be done in polynomial time.

    For every part of the regular expression $E$, we construct an automaton $G_i$ as follows.

    For \textbf{A}(1), an NFA $G_1$ can directly be constructed, as well as $G_2$ for \textbf{A}(2), cf. Fig.~\ref{fig2}.

    For \textbf{A}(3) and \textbf{A}(4), we make use of Lemma~\ref{lem1}. Namely, each of $G_3$ and $G_4$ is a parallel composition of $n$ automata, where the private alphabets are different. It means that we apply Lemma~\ref{lem1} always with a fresh alphabet $\Gamma$. Then, the alphabets of $G_3$ and $G_4$ are $\Delta\cup\Gamma_3$ and $\Delta\cup\Gamma_4$, where $\Gamma_3\cap\Gamma_4=\emptyset$. Notice that Lemma~\ref{lem1} can be applied to \textbf{A}(4), because $\Delta^{2^n} = \Delta \cdot \Delta^{2^n-1}$. Thus, $G_3$ for \textbf{A}(3) is a projection to $\Delta$ of the parallel composition of $n$ automata $B_n^3\|\cdots\|B_1^3$; similarly for $G_4$ for the regular expression \textbf{A}(4).

    In more detail, the language of $B_i^3$ is $\#\cdot L_m(A_i)\cdot \#\cdot \Delta^*$, where $A_i$ is as in Lemma~\ref{lem1} over the alphabet $\Sigma\cup\Gamma_3$, where $\Sigma=\Delta\cup\{e\}$ with $e$ being a new unobservable event. Then, by Lemma~\ref{lem1}, $P(L_m(B_n^3\|\cdots\|B_1^3))=\#\cdot (\Delta\cup\{\varepsilon\})^{2^n-1}\cdot \#\cdot \Delta^*$.

    For \textbf{B}, $G_5$ can easily be constructed.

    For \textbf{C}, we construct $|\Delta|^3$ automata $G_{c_1c_2c_3}$. Namely, for every $c_1c_2c_3$, we construct $G_{c_1c_2c_3}$ recognizing
    \[
      \Delta^* \cdot c_1c_2c_3 \cdot \Delta^{2^n-1} \cdot (\Delta\setminus N(c_1,c_2,c_3)) \cdot \Delta^*\,.
    \]
    Again, $G_{c_1c_2c_3}$ is a parallel composition of $n$ automata $C_i$ obtained from automata $A_i$ of Lemma~\ref{lem1} as follows. Every $C_i$ has a prefix recognizing $\Delta^* \cdot c_1c_2c_3$, and a suffix recognizing $(\Delta\setminus N(c_1,c_2,c_3)) \cdot \Delta^*$. Thus, the language of $C_i$ is
    \[
      \Delta^* \cdot c_1c_2c_3 \cdot L_m(A_i) \cdot (\Delta\setminus N(c_1,c_2,c_3)) \cdot \Delta^*\,.
    \]
    Then $L_m(G_{c_1c_2c_3}) = L_m(\|_{i=1}^{n} C_i)$.
    Let $P$ be a projection from the overall alphabet to $\Delta$. Then $P(L_m(G_{c_1c_2c_3})) = \Delta^* \cdot c_1c_2c_3 \cdot \Delta^{2^n-1} \cdot (\Delta\setminus N(c_1,c_2,c_3)) \cdot \Delta^*$.

    We now denote all the automata constructed above as $G_i$, for $i=1,\ldots,m$, where $m=|\Delta|^3+5$ is polynomial, and every $G_i$ was constructed in polynomial time. Then, by construction, we have that
    \[
      L(E)=\bigcup_{i=1}^{m} P(L_m(G_i))\,.
    \]

    In the above constructions, we use total automata (every automaton can be made total by adding a single state and the missing transitions), that is, in every state a transition under every event is defined. Then every $G_i$ is also total.

    Let $q_{s_i}$ and $q_{f_i}$, $i=1,\ldots,m$, be new non-marked states with self-loops under all events from $\Delta\cup\{\diamond\}$, where $\diamond$ is a new observable event.

    For every $i$, we modify $G_i$ by adding $q_{s_i}$ to the set of initial states, and by adding a transition from every state under the event $\diamond$ to state $q_{s_i}$, and from every marked state to state $q_{f_i}$, cf. Fig.~\ref{fig3}.
    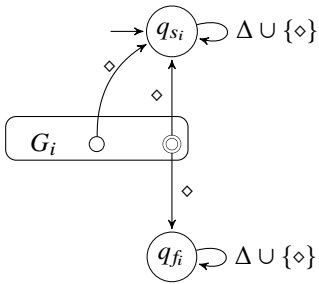
\begin{figure}[ht]
      \centering
      \begin{tikzpicture}[>=stealth',->,auto,shorten >=1pt,node distance=3cm,double distance=1pt,
        state/.style={circle,minimum size=0mm,inner sep=2pt,very thin,draw=black,initial text=}]

        \node [state,accepting]   (a1)  {};
        \node [state]             (d1) [left of=a1,node distance=1cm] {};
        \node []                  (a2) [left of=d1,node distance=.7cm]  {$G_i$};
        \node [state]             (2)  [below of=a1,node distance=1.5cm] {$q_{f_i}$};
        \node [state,initial]     (3)  [above of=a1,node distance=1.5cm] {$q_{s_i}$};

        \path
          (a1) edge node {$\diamond$} (2)
          (2)  edge[loop right] node {$\Delta\cup\{\diamond\}$} (2)
          (a1) edge node {$\diamond$} (3)
          (d1) edge[bend left] node[above,pos=.5] {$\diamond$} (3)
          (3)  edge[loop right] node {$\Delta\cup\{\diamond\}$} (3)
        ;

        \begin{pgfonlayer}{background}
          \path (a1.south  -| a2.west)+(-0.2,-0.1) node (a) {};
          \path (a2.north  -| a1.east)+(0.1,0.1)   node (b) {};
          \path[rounded corners, draw=black] (a) rectangle (b);
        \end{pgfonlayer}
      \end{tikzpicture}
      \caption{An illustration of the modification of $G_i$}
      \label{fig3}
    \end{figure}

    If $G_i$ is of the form $\|_{j=1}^{n} B_j^i$, we do this as follows. A transition under $\diamond$ to $b_{s_j}^i$ is added to every state of $B_j^i$, where $b_{s_j}^i$ is a new initial state of $B_j^i$, and for every marked state $t_j^i$ of $B_j^i$, we add the transitions $(t_j^i,\diamond_i,{t_j^i}')$ and $({t_j^i}', \diamond, b_{f_j}^i)$, where ${t_j^i}'$ and $b_{f_j}^i$ are new states added to $B_j^i$. Then, $q_{s_i}=(b_{s_1}^i,\ldots,b_{s_n}^i)$, $q_{f_i}=(b_{f_1}^i,\ldots,b_{f_n}^i)$, and for every state of $G_i$ of the form $(r_1,r_2,\ldots,r_n)$, a transition under $\diamond_i$ is defined in $G_i$ if and only if $(r_1,r_2,\ldots,r_n)$ is marked in $G_i$, that is, every $r_j$ is marked in $B_j^i$, cf. Fig.~\ref{fig4} for an illustration.
    \begin{figure}[ht]
      \centering
      \begin{tikzpicture}[>=stealth',->,auto,shorten >=1pt,node distance=3cm,double distance=1pt,
        state/.style={circle,minimum size=0mm,inner sep=2pt,very thin,draw=black,initial text=}]

        \node [state,accepting]   (a1)  {};
        \node [state]             (a3) [left of=a1,node distance=1cm] {};
        \node []                  (a2) [left of=a3,node distance=.7cm]  {$B_1^i$};
        \node [state]             (2a) [below of=a1,node distance=1.3cm] {};
        \node [state]             (2)  [below of=2a,node distance=1.3cm] {$b_{f_1}^i$};
        \node [state,initial]     (3)  [above of=a1,node distance=2cm] {$b_{s_1}^i$};

        \node [state,accepting]   (x1) [right of=a1,node distance=3cm] {};
        \node [state]             (x3) [right of=x1,node distance=1cm] {};
        \node []                  (x2) [right of=x3,node distance=.7cm]  {$B_2^i$};
        \node [state]             (z2a) [below of=x1,node distance=1.3cm] {};
        \node [state]             (z2) [below of=z2a,node distance=1.3cm] {$b_{f_2}^i$};
        \node [state,initial right]     (z3) [above of=x1,node distance=2cm] {$b_{s_2}^i$};

        \path
          (a1) edge node[left] {$\diamond_i$} (2a)
          (2a) edge node[left] {$\diamond$} (2)
          (2)  edge[loop left] node {$\Delta\cup\{\diamond\}$} (2)
          (a1) edge node {$\diamond$} (3)
          (a3) edge[bend left] node {$\diamond$} (3)
          (3)  edge[loop above] node {$\Delta\cup\{\diamond\}$} (3)
          (x1) edge node {$\diamond_i$} (z2a)
          (z2a) edge node {$\diamond$} (z2)
          (z2)  edge[loop right] node {$\Delta\cup\{\diamond\}$} (z2)
          (x1) edge node[right] {$\diamond$} (z3)
          (x3) edge[bend right] node[right] {$\diamond$} (z3)
          (z3)  edge[loop above] node {$\Delta\cup\{\diamond\}$} (z3)
          ;

        \begin{pgfonlayer}{background}
          \path (a2.south  -| a2.west)  node (a) {};
          \path (a2.north  -| a1.east)+(0.1,0)  node (b) {};
          \path (x2.south  -| x1.west)+(-0.1,0)  node (x) {};
          \path (x2.north  -| x2.east)  node (y) {};
          \path[rounded corners, draw=black] (a) rectangle (b);
          \path[rounded corners, draw=black] (x) rectangle (y);
        \end{pgfonlayer}
      \end{tikzpicture}
      \caption{An illustration of the modification of $G_i=B_1^i \parallel B_2^i$}
      \label{fig4}
    \end{figure}
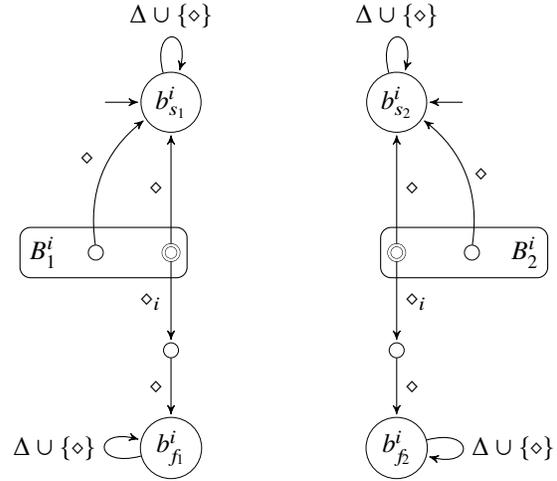

    In other words, every $G_i$ is always in state $q_{s_i}$. If a reachable state $(r_1,r_2,\ldots,r_n)$ of $G_i$ contains a non-marked state $r_j$ of $B_j^i$, then $(r_1,r_2,\ldots,r_n)$ leads only to state $q_{s_i}$ under $\diamond$ in $P(G_i)$. However, if $(r_1,r_2,\ldots,r_n)$ consists only of marked states of $B_1^i,\ldots,B_n^i$, then $(r_1,r_2,\ldots,r_n)$ leads to state $q_{s_i}$ under $\diamond$ and to state $(r_1',r_2',\ldots,r_n')$ under $\diamond_i$ in $G_i$, and hence to states $q_{s_i}$ and $q_{f_i}$ in $P(G_i)$.

    Summarized, if $w \in P(L_m(G_i))$ and $(r_1,r_2,\ldots,r_n)$ is a marked state reached by $w$, then
    \[
      (r_1,r_2,\ldots,r_n)\xrightarrow{\diamond} \{q_{s_i},q_{f_i}\}\,.
    \]
    If $w\notin P(L_m(G_i))$, then any state $(r_1,r_2,\ldots,r_n)$ reachable by $w$ is such that
    \[
      (r_1,r_2,\ldots,r_n)\xrightarrow{\diamond} \{q_{s_i}\}\,.
    \]

    We now consider the parallel composition $\|_{i=1}^{m} G_i$. Assume that the alphabet of $\|_{i=1}^{m} G_i$ is $\Sigma$. We show that $\|_{i=1}^{m} G_i$ is weakly (periodically) detectable with respect to $\Sigma \setminus (\Delta\cup\{\diamond\})$ if and only if $\M$ accepts $x$. In other words, if and only if $L(E)\neq\Delta^*$.

    Notice that $\|_{i=1}^{m} G_i$ is weakly (periodically) detectable with respect to $\Sigma \setminus (\Delta\cup\{\diamond\})$ if and only if $P(\|_{i=1}^{m} G_i)$ is weakly (periodically) detectable with respect to $\emptyset$.

    Assume first that $\M$ does not accept $x$. Then, for every $w\in\Delta^*$, $w$ is not an encoding of an accepting computation of $\M$ on $x$. Therefore, there exists $i\in\{1,\ldots,m\}$ such that $w\in P(L_m(G_i))$. This means that there is $w'\in P^{-1}(w)$ such that the composition $\|_{i=1}^{m} G_i$ is in a state $(x_1,x_2,\ldots,x_m)$ after reading $w'$, where $x_i$ is a marked state of $G_i$. This means that the system $P(\|_{i=1}^{m} G_i)$ after reading $w\diamond$ is in at least two states, namely $\bar q_s = (q_{s_1},\ldots,q_{s_i},\ldots,q_{s_m})$ and $(q_{s_1},\ldots,q_{s_{i-1}},q_{f_i},q_{s_{i+1}},\ldots,q_{s_m})$, and hence the system is not weakly (periodically) detectable.

    On the other hand, if $\M$ accepts $x$, then there is a word $w$ that encodes the accepting computation of $\M$ on $x$. Therefore, $w\notin L(E) = \bigcup_{i=1}^{m} P(L_m(G_i))$, and hence none of $G_i$ is in a marked state after reading any $w'\in P^{-1}(w)$. Since event $\diamond$ leads every non-marked state of $G_i$ only to state $q_{s_i}$, we have that $w\diamond$ leads $P(\|_{i=1}^{m} G_i)$ only to state $\bar q_s = (q_{s_1},\ldots,q_{s_m})$. Thus, $P(\|_{i=1}^{m} G_i)$ is weakly (periodically) detectable.
  \end{proof}

  In some cases, the projection erases only private events.
  This, for instance, ensures that the projection commutes with parallel composition, that is, $P(L_m(G_1\|G_2)) = P(L_m(G_1)) \parallel P(L_m(G_2))$. We now show that under this assumption, deciding weak modular detectability is a simpler problem.

  \begin{thm}\label{thm2}
    Let $\{G_1,G_2,\ldots,G_n\}$ be a set of discrete event systems and $P\colon\Sigma \to \Sigma_{o}$ be a projection such that all shared events of any two systems are included in $\Sigma_{o}$. Then deciding weak (periodic) modular detectability is PSPACE-complete.
  \end{thm}
  \begin{proof}
    We first show that $P(\|_{i=1}^{n} G_i)$ is weakly (periodically) detectable if and only if $\|_{i=1}^{n} P(G_i)$ is weakly (periodically) detectable. To this end, we show that a state is reachable in $P(\|_{i=1}^{n} G_i)$ by a word $P(w)$ if and only if it is reachable in $\|_{i=1}^{n} P(G_i)$ by $P(w)$.

    If a state $(x_1,\ldots,x_n)$ of $P(\|_{i=1}^{n} G_i)$ is reachable by a word $P(w)$, then every $x_i$ is reachable by $P(w)$ in $P(G_i)$. Therefore, state $(x_1,\ldots,x_n)$ is reachable by $P(w)$ in $\|_{i=1}^{n} P(G_i)$.

    On the other hand, if $(x_1,\ldots,x_n)$ is reachable by a word $P(w)$ in $\|_{i=1}^{n} P(G_i)$, then there are words $w_i$, $i=1,\ldots,n$, such that $P(w_i)=P(w)$, and $x_i$ is reachable by $w_i$ in $G_i$. Let $P(w)=a_1a_2\cdots a_m$, for some $m\ge 0$. Then every $w_i$ can be written as $w_i=u_{1,i} a_1 u_{2,i} \cdots u_{m-1,i} a_m u_{m,i}$, where every $u_{k,i} \in E_i^*$, $k=1,\ldots,m$, for $E_i$ denoting the private alphabet of $G_i$, that is, for any $i\neq j$, $E_i\cap E_j=\emptyset$. Then the word
    \[
      u_{1,1}u_{1,2}\cdots u_{1,n}\, a_1\, u_{2,1}\cdots u_{2,n} \cdots u_{m-1,n}\, a_m\, u_{m,1}\cdots u_{m,n}
    \]
    leads the parallel composition $\|_{i=1}^{n} G_i$ to state $(x_1,\ldots,x_n)$. Therefore, state $(x_1,\ldots,x_n)$ is reachable in $P(\|_{i=1}^{n} G_i)$ by $P(w)$.

    The fact that $P(\|_{i=1}^{n} G_i)$ is weakly (periodically) detectable if and only if $\|_{i=1}^{n} P(G_i)$ is means that instead of computing the observer of $\|_{i=1}^{n} G_i$, we can compute the parallel composition of observers of $G_i$.

    We now slightly abuse the notation and use $P(G_i)$ to denote the observer of $G_i$ with respect to $\Sigma_o$. Then we compute $\|_{i=1}^{n} P(G_i)$ on-the-fly using the nondeterministic search. The PSPACE algorithm needs to remember a state of $\|_{i=1}^{n} P(G_i)$, which consists of $n$ subsets each of size at most $k$, where $k$ is the maximum size of the state set among all $G_i$. Hence every state of $\|_{i=1}^{n} P(G_i)$ is of size at most $n\cdot k$, which is polynomial in the size of the input. To check that $\|_{i=1}^{n} P(G_i)$ is weakly detectable then means to guess a reachable state $(X_1,\ldots,X_n)$ of $\|_{i=1}^{n} P(G_i)$, where every $X_i$ is a singleton, such that the state is non-trivially reachable from itself by states consisting only of singletons~\cite{ShuLin2011}. (Similarly for weak periodic detectability.)
    Since PSPACE-hardness is known~\cite{YinLafortune17,Zhang17}, the problem is PSPACE-complete.
  \end{proof}

  Theorem~\ref{thm2} discusses a special case for which deciding weak modular detectability is easier. Consequently, if all events are observable, Theorem~\ref{thm2} implies that deciding weak modular detectability is PSPACE-complete. Therefore, unobservable shared events are essential in the proof of EXSPACE-hardness of deciding weak modular detectability.

\section{Modular Opacity}
Opacity is a property related to the privacy and security analysis of DES.
The system has a secret modeled as a set of states
and an intruder is modeled as a passive observer with limited observation.
The system is opaque if the intruder never knows for sure that the system is in a secret state,
i.e., the secret is not revealed.
We first recall the definition of current-state opacity~\cite{saboori2007notions}.

  \begin{defn}[Opacity~\cite{saboori2007notions}]
    A discrete event system $G = (Q, \Sigma, \delta, I)$ is current-state opaque with respect to $\Sigma_{uo}$ and a set of secret states $Q_S\subseteq Q$ if
      $(\forall s\in L(G)) [ R_G(s)\not\subseteq Q_S ]$.
  \end{defn}

  The modular version of the problem is defined as follows.
  \begin{defn}[Modular opacity]
    Given a set of discrete event systems $\{G_1,G_2,\ldots,G_n\}$, a set of unobservable events $\Sigma_{uo}$, and a set of secret states $Q_S$. The modular opacity problem asks whether the system $G_1\|G_2\| \cdots \| G_n$ is opaque with respect to $\Sigma_{uo}$ and $Q_S$.
  \end{defn}

  We show that deciding whether a modular system is opaque requires exponential space.
  \begin{thm}
    Deciding modular opacity is an EXPSPACE-complete problem.
  \end{thm}
  \begin{proof}
    Checking opacity for a monolithic system is done by constructing the observer and by checking whether there is a state in the observer that is a subset of the secret states~\cite{saboori2007notions}.
    Since, in the modular setting, the observer can be constructed nondeterministically in exponential space, the problem is in EXPSPACE for modular systems.

    To prove EXPSPACE-hardness, we reuse the proof of Theorem~\ref{thm1} and show that the modular system constructed there is opaque if and only if it is not weakly detectable.
    Let $G_i$, $i=1,\ldots,m$, be the automata constructed in the proof of Theorem~\ref{thm1},
    and let $Q_S=\{\bar{q}_s\}$ be the set of secret states.

    As shown in the proof of Theorem~\ref{thm1},
    if $\|_{i=1}^m G_{i}$ is weakly detectable with respect to $\Sigma\setminus(\Delta\cup\{\diamond\})$,
    then there exists a word $w$ such that, after reading $w\diamond$,
    the observer with unobservable events $\Sigma\setminus(\Delta\cup\{\diamond\})$
    knows for sure that the modular system is in state $\bar{q}_s$.
    This means that the system is not opaque with respect to $\{\bar{q}_s\}$ and $\Sigma\setminus(\Delta\cup\{\diamond\})$, since the unique secret state is revealed.

    On the other hand, if $\|_{i=1}^mG_{i}$ is not weakly detectable with respect to $\Sigma\setminus(\Delta\cup\{\diamond\})$,
    then, after reading any word, we cannot distinguish state $\bar{q}_s$  from some other state,
    and hence the system is opaque with respect to $\{\bar{q}_s\}$ and $\Sigma\setminus(\Delta\cup\{\diamond\})$.
  \end{proof}

  A similar special case to Theorem~\ref{thm2} can be shown for modular opacity.
  \begin{thm}\label{thm3}
    Let $\{G_1,G_2,\ldots,G_n\}$ be a set of discrete event systems and $P\colon\Sigma \to \Sigma_o$ be a projection such that all shared events of any two systems are included in $\Sigma_o$. Then deciding modular opacity is PSPACE-complete.
  \end{thm}
  \begin{proof}
    If the projection erases only private events, we have that $P(\|_{i=1}^{n} G_i)$ is opaque if and only if $\|_{i=1}^{n} P(G_i)$ is opaque. The proof is similar to that in the proof of Theorem~\ref{thm2}. Abusing the notation by denoting the observer of $G_i$ as $P(G_i)$, we obtain that deciding opacity for $\|_{i=1}^{n} P(G_i)$ can be done in PSPACE using the nondeterministic search method, since every state of $\|_{i=1}^{n} P(G_i)$ consists of $n$ subsets each of size at most $k$, where $k$ is the maximum number of states among all $G_i$, which is polynomial in the size of the input. 
    More specifically, to check that the system is not opaque, we guess a reachable state $(X_1,\ldots,X_n)$ of $\|_{i=1}^{n} P(G_i)$ such that $X_1\times\cdots\times X_n$ is included in the set of secret states $Q_S$. Since PSPACE is closed under complement, checking opacity is also in PSPACE. Moreover, since PSPACE-hardness is known~\cite{CassezDM12}, the problem is PSPACE-complete.
  \end{proof}

\section{Modular A-Diagnosability}

Fault diagnosis is an important task in discrete event systems.
Let $\Sigma_F\subseteq\Sigma$ be a set of fault events,
and let $L_F=\Sigma^*\Sigma_F\Sigma^*$ be the set of all trajectories that contain a fault event.
The goal of fault diagnosis is to detect the occurrence of fault events.
To this end, several different notions of diagnosability have been proposed in the literature.
For example, the diagnosability problem of Sampath et al.~\cite{sampath1995diagnosability} requires that the occurrence of a fault can always be detected within a finite delay.
This problem is polynomial-time decidable for monolithic systems~\cite{jiang2001polynomial,yoo2002polynomial}
and PSPACE-complete for modular systems~\cite{YinLafortune17}.

Thorsley and Teneketzis~\cite{thorsley2005diagnosability} proposed a weaker version of diagnosability, so-called A-diagnosability.
Compared to diagnosability, where a fault has to be detected on every path after it happens within a finite delay, A-diagnosability requires that for any  path that contains a fault event,
it always has an extension where a fault can be detected.

\begin{defn}[A-diagnosability~\cite{thorsley2005diagnosability}]\label{def:adiag}
  A discrete event system $G = (Q, \Sigma, \delta, I)$ is A-diagnosable with respect to
  $\Sigma_{uo}$ and $\Sigma_F$ if for any fault trajectory, there exists an extension under which a fault event has occurred, that is,
    $(\forall s\in L(G)\cap L_F)
    (\exists t\in L(G)/s)
    [ P^{-1}P(st) \cap L(G)\subseteq L_F ]$,
  where $L(G)/s = \{t\in\Sigma^* \mid st \in L(G)\}$.
\end{defn}

  Intuitively, A-diagnosability says that
  after the occurrence of any fault,
  we can detect the occurrence of a fault  \emph{with probability one}. That is,
  if we assume that each transition is assigned a non-negative transition probability, then the probability of detecting the fault converges to one as the length of the trajectory increases.
  In fact, A-diagnosability is originally defined for stochastic discrete event systems.
  Formally, let $p\colon Q\times\Sigma\times Q\to[0,1]$ be a transition probability function that assigns to each transition of $G$ a non-negative probability. Then the originally definition of A-diagnosability~\cite{thorsley2005diagnosability} requires that
  $(\forall \epsilon>0)(\exists N\in\mathbb{N})(\forall s\in L(G)\cap L_F)(\forall n\geq N)
  [Pr(t: P^{-1}P(st) \cap L(G)\not\subseteq L_F\mid t\in L(G)/s\wedge |t|=n)<\epsilon]$,
  where $Pr(\cdot)$ above denotes the  probability of occurrence of all continuations of $s$ with length $n$ under which a fault event cannot be detected unambiguously.
  However, this definition does not depend on the transition probability function $p$,
  since for any DES $G$ satisfying Definition~\ref{def:adiag},
  term $Pr(\cdot)$ is always decreasing as $n$ increases --
  the transition probability function $p$ only affects the decreasing rate, i.e.,
  for each $\epsilon$, the specific value of $N$ may be affected but the existence of such an integer is not.
  Therefore, we use an equivalent logical definition here to simplify our proof.
  
  The modular version of the problem is defined as follows.
  \begin{defn}[Modular A-diagnosability]
    Given a set of discrete event systems $\{G_1,G_2,\ldots,G_n\}$, a set of unobservable events $\Sigma_{uo}$, and a set of fault events $\Sigma_F$. The modular A-diagnosability problem asks whether the discrete event system $G_1\|G_2\| \cdots \| G_n$ is A-diagnosable with respect to $\Sigma_{uo}$ and $\Sigma_F$.
  \end{defn}

Bertrand et al.~\cite{bertrand2014foundation} and Chen et al.~\cite{chenrevised} have shown that testing A-diagnosability is PSPACE-complete for monolithic systems. Hereafter, we show that this problem is EXPSPACE-complete for modular systems.
\begin{thm}\label{a-diag-thm}
  Deciding modular A-diagnosability is EXPSPACE-complete.
\end{thm}
\begin{proof}
Checking A-diagnosability for a monolithic system can be done in PSPACE~\cite{bertrand2014foundation,chenrevised}, and hence EXPSPACE is sufficient for modular systems.

We now show that the problem is EXPSPACE-hard. To this aim, we reuse and slightly modify the proof of Theorem~\ref{thm1} to show that a $2^n$-space-bounded Turing machine $\M$ accepts an input $x$ of length $n$ if and only if the modular system we construct is A-diagnosable.

Let $G_i$ be the automata constructed in the proof of Theorem~\ref{thm1}.
Let $f$ be a new unobservable event, which is the sole fault event. Let $\Box$ be a new observable event, and for $i=1,\ldots,m$, let $\Box_i$ be a new unobservable event.

We further modify every $G_i$ by adding the transitions
\begin{enumerate}
  \item $(q_{s_i},f,q_{s_i}')$, where $q_{s_i}'$ is a new state,
  \item $(q_{s_i}',\Box,I_i)$, where $I_i$ are the initial states of $G_i$,
  \item $(q_{s_i},\Box_j,q_{s_i}')$, for  $j\neq i$, and
  \item $(q_{f_i},\Box_j,q_{s_i}')$, for  $j=1,\ldots,m$.
\end{enumerate}

Intuitively, having read a word $w\diamond$, for some $w\in\Delta^*$, the modular system is in state $\bar{q}_s=(q_{s_1},q_{s_2},\dots,q_{s_m})$ and perhaps also in a state $q=(\ldots,q_{f_i},\ldots)$. From state $\bar{q}_s$, it can go to state $\bar{q}_s'=(q_{s_1}',q_{s_2}',\dots,q_{s_m}')$ under $f$, where the fault occurs. But the system can also go to state $\bar{q}_s$ from $q$ under $\Box_i$, which is a word with the same projection but without the fault event $f$. Notice that $\Box_i$ is possible only if there is state $q_{f_i}$, for some $i$, indicating that $G_i$ marks word $w$.
This mechanism prevents $q$ from going to state $\bar{q}_s'$ if none of $G_i$ marks $w$.
A conceptual illustration is provided in Fig.~\ref{fig_i}.
  \begin{figure}[ht]
    \centering
      \begin{tikzpicture}[>=stealth',auto,baseline,->,shorten >=1pt,node distance=2.5cm, state/.style={ellipse,minimum size=7mm,very thin,draw=black,initial text=}]
      \node[state]            (1p) {initial};
      \node[state]            (1r) [below right of=1p] {$\bar{q}_s$};
      \node[state]            (11) [below left of=1r] {$\bar{q}_s'$};
      \node[state]            (1s) [below left of=10] {$q$};

      \path
        (1s) edge node {$\Box_i$} (11)
        (1r) edge node {$f$} (11)
        (11) edge[out=90,in=270] node {$\Box$} (1p)
        (1p) edge node[above left] {$w_2\diamond$} (1s)
        (1p) edge node {$w_1\diamond$} (1r)
        ;
      \end{tikzpicture}
      \caption{An illustration of the proof of Theorem~\ref{a-diag-thm}}
      \label{fig_i}
  \end{figure}
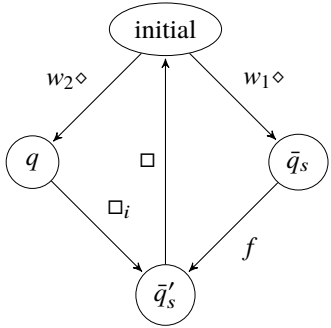

Hereafter, we show that $\M$ accepts $x$ if and only if $\|_{i=1}^mG_i$ is A-diagnosable with respect to fault events $\{f\}$ and unobservable events $\Sigma\setminus (\Delta\cup\{\diamond,\Box\})$.

Suppose that $\M$ does not accept $x$.
Then, as shown in the proof of Theorem~\ref{thm1}, for every word $w\in\Delta^*$, $P(w)\diamond$ ends up in at least two states, namely $\bar{q}_s$ and states of the form $q=(\ldots,q_{f_i},\ldots)$, where ``$\ldots$'' are either $q_{f_j}$ or $q_{s_j}$. Notice that there is an unobservable transition from $q$ under $\Box_i$ to $\bar{q}_s'=(q_{s_1}',\ldots,q_{s_m}')$.  That is, there are $w_1$ and $w_2$ such that $P(w_1)=P(w_2)=P(w)$ and $w_1\diamond$ leads to state $\bar{q}_s$ and $w_2\diamond$ to state $q$. Then $w_1\diamond f$ containing the fault event and $w_2 \diamond \Box_i \in P^{-1}P(w_1\diamond f)$ both end up in state $q_s'$. Then, only event $\Box$ is possible, which leads the system to the initial states. The proof now follows by induction. Therefore, for any extension of $w_1\diamond f$, there is always a path that bypasses the fault event $f$ due to some event $\Box_k$, and hence $P^{-1}P(w_1\diamond fw') \cap L(G)\not\subseteq L_F$ for any extension $w'$ of $w_1\diamond f$ in the system. Therefore, the system is not A-diagnosable.

On the other hand, if $\M$ accepts $x$, then there exists a word $w$ such that, observing $P(w)\diamond$, the system reaches only states $\bar{q}_s$ and $\bar{q}_s'$.
Let $s\in L(\|_{i=1}^mG_i)\cap L_F$ be a fault trajectory.
We now fix a word $t'$ such that $st'\Box$ is defined in the system; notice that such a word exists.
Then $t=t' \Box w\diamond f \Box\in L(\|_{i=1}^{m} G_i)$ is an extension of $s$ such that all words of $L(G)$ with the projection equal to $P(st)$ contain the fault event $f$. Indeed, after observing $P(st')\Box$, the observer of the system is in the initial state, and $P(w)\diamond$ leads the observer only to state $\{\bar{q}_s,\bar{q}_s'\}$. Now, when the observer sees event $\Box$, it is sure that the fault event $f$ has occurred, that is, $P^{-1}P(st) \cap L(G)\subseteq L_F$.
Since the fault trajectory $s$ was chosen arbitrarily, the system is A-diagnosable.
\end{proof}

A special case similar to that of Theorem~\ref{thm2} also holds for A-diagnosability.
\begin{thm}\label{thm20}
Let $\{G_1,G_2,\ldots,G_n\}$ be a set of discrete event systems and $P\colon\Sigma \to \Sigma_o$ be a projection such that all shared events of any two systems are included in $\Sigma_o$. Then deciding modular A-diagnosability is PSPACE-complete.
\end{thm}
\begin{proof}
Since the problem is trivial for observable fault events, we only consider unobservable fault events. Then the assumption that all shared events are observable implies that
all fault events are private, that is, each fault event can only occur locally.
Let $\Sigma_{F_i} = \Sigma_F \cap \Sigma_i$ be the set of all local fault events, $i=1,\ldots,n$.

Notice that A-diagnosability is an event-based property and that fault events are erased in $P(G_i)$.
To address this issue, we reformulate A-diagnosability as a state-based property.
To this end, we assume, without loss of generality, that for each $G_i$,
its state-space $Q_i$ is partitioned as $Q_i=Q_{N_i} \dot{\cup}\, Q_{F_i}$ so that the system is in states of $Q_{N_i}$ as long as no fault has occurred and it is in states of $Q_{F_i}$ from the moment a fault has occurred on. This assumption can be easily fulfilled in polynomial-time by computing the product of $G_i$ with a two-state deterministic automaton marking the language $L_F$.
In this setting, A-diagnosability can be reformulated as follows:
  $(\forall s\in L(G):\delta(I,s)\subseteq Q_F)(\exists t\in L(G)/s) [R_G(st) \subseteq Q_F]$,
where $Q_F=\bigcup_{i=1}^{n} Q_1 \times\cdots\times Q_{i-1} \times Q_{F_i} \times Q_{i+1} \times\cdots\times Q_{n}$ is the set of all states where at least one component indicates that a fault event has occurred.

If the projection erases only private events, we have shown in the proof of Theorem~\ref{thm2} that a state is reachable in $P(\|_{i=1}^{n} G_i)$ by a word $P(w)$ if and only if it is reachable in $\|_{i=1}^{n} P(G_i)$ by $P(w)$. Therefore, to verify A-diagnosability of the modular system $\|_{i=1}^{n} G_i$ with respect to the set of fault events $\Sigma_F$, rather than to compute the observer of the modular system, we compute the composition of observers of $G_i$, denoted $\|_{i=1}^{n} P(G_i)$ in the rest of this proof.
Deciding A-diagnosability for $\|_{i=1}^{n} P(G_i)$ can be done in PSPACE by using the nondeterministic search method, since every state in the product of observers $\|_{i=1}^{n} P(G_i)$ consists of at most $n\cdot k$ states, where $k$ is the maximum of states among all $G_i$, which is polynomial in the size of the input.
To verify A-diagnosability then means to check that for every reachable state $(X_1,\ldots,X_n)$ of $\|_{i=1}^{n} P(G_i)$, if there is $i$ such that $X_i\cap Q_{F_i} \neq \emptyset$, then there is a state $(Y_1,\ldots,Y_n)$ reachable from $(X_1,\ldots,X_n)$ with the property that there is $j$ such that $Y_j \subseteq Q_{F_j}$.
Since PSPACE-hardness is known~\cite{bertrand2014foundation,chenrevised}, the problem is PSPACE-complete.
\end{proof}

\section{Conclusion}
In this paper, we showed that deciding weak detectability, opacity, and A-diagnosability for modular discrete event systems are EXPSPACE-complete problems.
Our results reveal that these properties are significantly more difficult to verify in the modular setting compared with their monolithic counterparts.
Special cases were identified where these properties can be checked in polynomial space.
Our results also reveal the connections and similarities among these properties from the structural point of view. 
\ack{The research of T. Masopust was supported by the GA\v{C}R project GA15-02532S and by RVO 67985840.}

\bibliographystyle{plain}
\bibliography{biblio}

\end{document}